\documentclass[10pt,transaction,twocolumn]{IEEEtran}
\usepackage{graphics}
\usepackage{graphicx}
\usepackage{amsmath,amsfonts,amssymb,varioref,mathrsfs,verbatim,amsthm}
\usepackage{latexsym}
\usepackage{epsfig}
\usepackage[nospace,noadjust]{cite}
\usepackage{latexsym}
\usepackage{epsfig}
\usepackage{epstopdf}
\usepackage[nospace,noadjust]{cite}
\usepackage{array}
\usepackage[english]{babel}
\usepackage{color}
\usepackage{hyperref}
\usepackage{babel}
\usepackage{algorithm,algpseudocode}% http://ctan.org/pkg/algorithms
\usepackage{xcolor,colortbl}
\usepackage{gensymb}
\newtheorem{theorem}{{\bf Theorem}}

\newtheorem{lemma}{\noindent {\bf Lemma}}

\include{new_commands}

\makeatletter
\let\l@ENGLISH\l@english
\makeatother

\DeclareMathOperator*{\argmin}{argmin}

\begin{document}%
\pagenumbering{gobble}
\title{Fluctuations of  the SNR at the output of the MVDR with Regularized Tyler Estimators}
\author{Khalil Elkhalil,~\IEEEmembership{Student Member,~IEEE,} Abla Kammoun,~\IEEEmembership{Member,~IEEE}, Tareq~Y.~Al-Naffouri,~\IEEEmembership{Member,~IEEE}, and Mohamed-Slim Alouini,~\IEEEmembership{Fellow,~IEEE}

\thanks{This work was funded by a CRG3 grant ORS\#2221 from the Office of Competitive
Research (OCRF) at King Abdullah University of Science and Technology
(KAUST). \par
The authors are with the Electrical Engineering Program, King Abdullah University of Science and Technology, Thuwal, Saudi Arabia; e-mails: \{khalil.elkhalil, abla.kammoun, tareq.alnaffouri, slim.alouini\}@kaust.edu.sa.
}
%\thanks{Mohammed Eltayeb and Hamid Reza Bahrami are with the Department
%of Electrical and Computer Engineering, The University of Akron, Ohio, USA; e-mails: \{me33, hrb\}@uakron.edu. Khalil Elkhalil and Tareq Y. Al-Naffouri are with the Electrical Engineering Department, King Abdullah University of Science and Technology, Thuwal, Saudi Arabia; e-mails: \{khalil.elkhalil, tareq.alnaffouri\}@kaust.edu.sa. Tareq Y. Al-Naffouri is also associated with the Department of Electrical Engineering, King Fahd University of Petroleum
%and Minerals, Dhahran 31261, Kingdom of Saudi Arabia.}% <-this % stops 
}
\maketitle
%
%\author{\authorblockN{Mohammed E. Eltayeb$^1$,  Hamid Reza Bahrami$^1$, and Tareq Y. Al-Naffouri$^{2,3}$}\\
%\authorblockA{$^1$Department of Electrical and Computer Engineering, The University of Akron, Ohio, USA\\
%$^2$Electrical Engineering Department, King Abdullah University of Science and Technology, Thuwal, Saudi Arabia\\ $^3$Electrical Engineering Department, King Fahd University of Petroleum
%and Minerals, Dhahran, Saudi Arabia \\ Emails: me33,hrb@uakron.edu, tareq.alnaffouri@kaust.edu.sa\ \ }}
\begin{abstract}
    This paper analyzes the statistical properties of the signal-to-noise ratio (SNR) at the output of the Capon's minimum variance distortionless response (MVDR) beamformers when operating over impulsive noises. Particularly, we consider the supervised case in which the receiver employs the regularized Tyler estimator in order to estimate the  covariance matrix of the interference-plus-noise process using $n$ observations of size $N\times 1$. %estimates the noise  covariance matrix, using $n$  observations of size $N\times 1$ of the interference-plus-noise process, by the regularized-Tyler estimator (RTE). 
    The choice for the regularized Tylor estimator (RTE) is motivated by its resilience to the presence of outliers and its regularization
    parameter that guarantees a good conditioning of the covariance estimate. Of particular interest in this paper is the derivation of the second order statistics of the SINR. To achieve this goal, we consider two different approaches. The first one is based on considering the classical regime, referred to as the $n$-large regime, in which $N$ is assumed to be fixed while $n$ grows to infinity. The second approach is built upon recent results developped within the framework of
    random matrix theory and assumes that $N$ and $n$ grow large together. 
    %In order to facilitate the computation of the fluctuations, we consider two asymptotic regimes: the first one is  the classical regime, referred to as the $n$-large regime, in which $N$ is fixed while $n$ grow to infinity.  The second regime, which is frequently within the framework of random matrix theory, assumes that $N$ and $n$ grow large together. These assumptions allow to leverage recent results
%    regarding the asymptotic behaviour of the RTE.
    %This allows us to leverage very recent results regarding the
%    behaviour of the RTE. 
    Numerical results are provided in order to compare between the accuracies of each regime under different settings.    %when the noise covariance matrix is estimated using implemented using the regularized Tyler estimator (RTE) for covariance matrix estimation. In particular, we establish a central limit theorem (CLT) on the output SNR of the MVDR filter under the supervised training case. Two asymptotic regimes are considered, the first is the classical regime that assumes a large number of snapshots, the second is the random matrix theory (RMT) regime, where both the number of snapshots and the number of antennas are large growing at the same rate. Numerical results are provided to validate the accuracy of the CLT and to determine the scenarios in which one regime is more accurate than the other.

\end{abstract}

%

%Relay selection is a simple technique that achieves spatial diversity in cooperative relay networks. However, for relay selection algorithms to make a selection decision, channel state information (CSI) from all cooperating relays is usually required at a central node. This requirement poses two important challenges which have been unjustly neglected in the literature. Firstly, the fed back channel information is usually corrupted by additive noise. This could lead to a transmission outage if the central node selects the set of cooperating relays based on inaccurate feedback information.  Secondly, CSI acquisition in the presence of noise in the feedback channel  generates a great deal of feedback overhead (air-time) that could result in significant transmission delays. In this paper, we introduce a limited feedback relay selection algorithm that is based on the theory of compressive sensing (CS). The proposed algorithm exploits CS to first obtain the identity of the strong relays. Following that, the CSI of the strong relays is estimated using minimum mean square error estimation without any additional feedback. To minimize the effect of noise on the fed back CSI, we introduce a back-off strategy that optimally  backs-off on the noisy received CSI. Numerical results show that the proposed algorithm drastically reduces the feedback air-time when compared to the well known timer algorithm  and achieves a rate close to that obtained by selection algorithms with dedicated error-free feedback channel.
\section{Introduction}
\PARstart{T}{he} minimum variance distortionless response (MVDR) beamformer or the Capon's MVDR beamformer is widely used in sensor array signal processing applications such as the inspection of direction of arrival (DOA) and the estimation of the power of a given signal of interest (SOI) \cite{Trees,Stoica}. 
The design of the MVDR beamforming requires the receiver to acquire an estimate of the unknown interference and noise covariance matrix.  Several covariance estimators constructed from  signal-free observations can be employed. The most popular ones are those based on the sample covariance matrix (SCM). Their popularity owe to  their low-complexity and the existence of a good understanding of their behaviour. However, SCM based estimators are well-known
to exhibit poor performances when observations contain outliers. This drawback becomes more acute in many applications such as radar and sonar processing where the noise is known to present an impulsive behaviour \cite{Ward,watts,haykin}. % among which we distinguish the sample covariance matrix (SCM) which is well-appreciated for its low complexity and the existence of a good understanding of its behaviour. However, the use of the SCM can be inadequate in practice if some atypical observations exhibit an impulsive character \cite{Ward,watts,haykin}. 
A promising alternative to the use of these estimators is represented by the class of robust scatter estimators. The latter can be traced back to  the early works of Huber \cite{huber-81} and Maronna \cite{maronna-76} in the seventies.  With  the emergence of new tools allowing the understanding of many robust covariance estimators, there is today a rekindled interest in the analysis of these estimators. The focus of our work is on the regularized tyler
estimator (RTE) for which a new wave of important results have been obtained  \cite{abla,couillet-Mckay,abla_romain}. 
% Research in the field of robust estimation is now knowing a growing interest , mainly because the rise  of new tools allowing the understanding of their behaviour.
%, however  has slowly lost momentum  over the subsequent decades, is now witnessing a growing interest after the  
%Of interest in this paper is the regularized Tyler estimator (RTE) that has recently been proposed  by \cite{Pascal-2013} and for which new asymptotic results of special interest have been reached \cite{abla,couillet-Mckay,couillet-kammoun-14}. A further important advantage of the RTE lies in its improved conditioning, as a result of its use of a regularization parameter that by construction prevents the covariance matrix from being singular. This feature becomes of particular interest in MVDR beamforming since it involves the computation of the inverse of the covariance estimator. \par
Of particular interest in this work is the behavior of the SNR at the output of the MVDR filter in impulsive noise environments. This question has not been addressed before. To the best of our knowledge, the existing works have thus far focused on the behavior of the SNR of the MVDR filter in Gaussian noise environments. %thus far been studied only when SCM based estimators are employed.
In this course, the early results have provided  an asymptotic characterization of the SNR  in the limiting regime defined by both the number of samples and their dimensions growing large with the same pace \cite{mestre}. Subsequently,  a second order analysis gaining insights into the fluctuations of the SNR has been carried out in \cite{rubio}. The objective of this paper is to extend the aforementioned results concerning the behaviour of the SNR at the output of the MVDR when
the noise sample covariance matrix is estimated using the RTE. Under this setting, we establish in this paper a central limit theorem (CLT) of the SNR under two different regimes.   The first regime corresponds to the classical one obtained by fixing the dimension of the observations and tending their number to infinity. The second regime, on the other hand, merely consists in assuming that the
number of samples and their dimensions grow large at the same pace. 
Both regimes can be of practical interest. Intuitively, the former is suitable for scenarios in which the number of observations is much greater than the array size, while the second is expected to be more accurate when the number of samples and their dimensions are of the same order of magnitude. Such intuition will be confirmed by a set of numerical results, comparing the performance of both regimes in terms of
some meaningful metrics. 

The remainder of this paper is organized as follows. Section II reviews  the MVDR beamforming and the RTE of the covariance matrix. In section III, the main result deriving the CLT of the SNR at the MVDR beamformer is provided. Prior to concluding, the last section presents a set of numerical results allowing to compare between both regimes. %Section II provides some background of this research. In section III, the main result related to the central limit theorem (CLT) on the the output SNR of the MVDR beamformer based on the use of the RTE is derived. Section IV provides some performance analysis and numerical results that validates our theoretical findings. Finally, some conclusions are drawn in section V. \\

\emph{Notations:} 
Throughout this paper, we use the following notations :  $\left(.\right)^t$, $\left(.\right)^*$ and $\textnormal{tr}\left(.\right)$ respectively stand for the transpose, the hermitian and the trace of a matrix. Also, $\otimes$ denotes the Kronecker product between two matrices while $\Re(.)$ and $\Im(.)$ respectively denote the real and the imaginary parts of a matrix.

\section{Robust MVDR Filtering with RTE}
\subsection{Robust MVDR}
We consider a uniform linear array (ULA) with $N$ sensors, receiving a narrow band source signal. The received vector at time $t$ can be represented by
$$
{\bf y}(t)={\bf s}_0 s(t) +{\bf x}(t),
$$
where ${\bf s}_0$ and $s(t)$ refer respectively to the array steering vector and the source signal at time $t$, whereas ${\bf x}(t)$ stands for the additive noise vector at time $t$. We assume that the distribution of the noise is heavy-tailed  belonging to the family of  compound-Gaussian distributions, i.e, ${\bf x}(t)$ can be put in the following form:
\begin{equation}
    {\bf x}(t)=\sqrt{\tau(t)}\boldsymbol{\Sigma}_N^{\frac{1}{2}}{\bf w}(t),
\label{eq:noise}
\end{equation}
where ${\bf w}(t)$ is a $N\times 1$ standard Gaussian vector, $\tau(t)$ is a positive random scalar called texture. Usually, $\tau(t)$ is drawn from heavy-tailed distribution in order to account for the impulsive character of the noise. Matrix $\boldsymbol{\Sigma}_N$ is the noise covariance matrix and is assumed to take the following form \cite{DOA}
%with a heavy tailed distribution modeling the impulsive character of the noise and $\boldsymbol{\Sigma}_N$ is the noise covariance matrix given by \cite{DOA}
\begin{equation}
\mathbf{\Sigma}_N = \sigma_0^2 \mathbf{I}_N + \sum_{i=1}^q\sigma_i^2 \mathbf{a}\left(\theta_i\right)\mathbf{a}\left(\theta_i\right)^*,
\label{eq:sigma}
\end{equation}
where $q$ is the number of interferers,  $\left\{\theta_i\right\}$, $i \in \left \{1, \cdots,q  \right \}$ are their corresponding angles of arrival, and ${\bf a}(\theta)$ is the $N\times 1$ array steering vector given by:
$$
\left[{\bf a}(\theta)\right]_k=\exp(\jmath 2\pi (k-1)\theta)
$$
%denote the directions of arrival (DOA) of $q$ interferences and $\mathbf{a}(.)$ is a $N \times 1$ steering vector.
The received vector is processed by a beamformer in order to enhance the desired signal while reducing the impact of the noise: 
$$
{ z}= {\bf u}^{*}{\bf y}(t).
$$
We consider in this paper the MVDR beamformer which seeks the best filter ${\bf u}$ that minimizes the power of the resulting noise while ensuring the distortionless response of the beamformer towards the direction of the desired source. The corresponding optimization problem is thus given by \cite{DOA}: 
\begin{equation} \label{MVDRoptim}
\mathbf{u}_0 = \argmin_{\mathbf{u}\in \mathbb{C}^{N\times 1}: \mathbf{u}^* \mathbf{s}_0=1} \mathbf{u}^* \mathbf{\Sigma}_N\mathbf{u},
\end{equation}
 Using the Lagrange method, it can be shown that $\mathbf{u}_0$ has the following closed-form expression:
\begin{equation}
\mathbf{u}_0 = \lambda \mathbf{\Sigma}^{-1}_N \mathbf{s}_0,
\label{eq:optimal_z0}
\end{equation}
where $\lambda$ is the Lagrange multiplier satisfying $\lambda = \frac{1}{\mathbf{s}_0^* \mathbf{\Sigma}^{-1}_N \mathbf{s}_0}$. 

As shown by \eqref{eq:optimal_z0}, the design of the MVDR beamforming requires the knowledge of the noise covariance matrix. In practice, this unknown covariance matrix is replaced by an estimate that is built from signal-free observations. In order to ensure the robustness of the beamformer towards the impulsive character of the noise, robust covariance estimators should be used.  In this paper we focus on the use of the regularized Tyler estimator (RTE).%we consider in this paper the use of the Regularized Tyler Estimator (RTE). 
%Since the traditional sample covariance matrix is known to be vulnerable to outliers, the focus in this paper is on the use of robust covariance estimators instead.  In particular, we consider the use of the (RTE) which will be introduced in the following section. 

%However in most practical scenarios, the covariance matrix is unknown and thus has to be estimated.
%We consider a signal impinging into a uniform linear array (ULA) with $N$ antennas. 
%We consider a signal impinging into a uniform linear array (ULA) 
\subsection{MVDR Beamforming Based on the RTE}
We assume that the receiver has previously acquired $n$ free source signal observations ${\bf x}_1,\cdots, {\bf x}_n$ drawn from the same distribution of ${\bf x}(t)$ in \eqref{eq:noise}, i.e,
%Consider $\mathbf{x}_1, \cdots, \mathbf{x}_n$, $n$ observations of size $N$, where 
\begin{equation}
\mathbf{x}_i = \sqrt{\tau_i}\mathbf{\Sigma}_N^{\frac{1}{2}} \mathbf{w}_i, \: i=1,\cdots,n.
\end{equation}
%where $\tau_i$ is a positive random variable independent from $\mathbf{w}_i$ with heavy tailed distribution modeling the impulsive nature in the received signal with $\mathbf{E}\left[\tau_i\right]=1$, $\mathbf{w}_i \sim \mathcal{CN}\left(\mathbf{0},\mathbf{I}_N\right)$ and $\mathbf{\Sigma}_N \succeq 0$ is the population covariance matrix such that $\frac{1}{N}\textnormal{tr} \mathbf{\Sigma}_N=1$ \cite{abla}.
 The regularized robust scatter estimator is defined as the unique solution $\widehat{\mathbf{C}}_N (\rho)$ to the following fixed-point equation:
\begin{equation} \label{RTE}
\widehat{\mathbf{C}}_N (\rho) = \left(1-\rho\right)\frac{1}{n}\sum_{i=1}^n \frac{\mathbf{x}_i\mathbf{x}_i^*}{\frac{1}{N}\mathbf{x}_i^*\widehat{\mathbf{C}}_N^{-1} (\rho)\mathbf{x}_i}+\rho \mathbf{I}_N,
\end{equation} 
where $\rho \in \left(\max\left(0,1-\frac{n}{N}\right),1\right]$ is the regularization parameter\footnote{The existence and uniqueness of $\widehat{\bf C}_N(\rho)$ is proved in \cite{Pascal-2013}.}. Note that the robustness of the RTE can be easily seen from \eqref{RTE} which reveals its invariance towards the scaling of ${\bf x}_i$ thus allowing the cancelling-out of the impact of $\tau_i$. 
%From (\ref{RTE}), the RTE is clearly robust against the impulsive r.v $\tau_i$ since its contribution is canceled. 
%\subsection{MVDR Beamformer}
%The optimal MVDR beamforming vector also known as the capon's MVDR filter is the solution of the following constrained minimization problem
%\begin{equation} \label{MVDRoptim}
%\mathbf{w}_0 = \argmin_{\mathbf{w}\in \mathbb{C}^{N\times 1}: \mathbf{w}^* \mathbf{s}_0=1} \mathbf{w}^* \mathbf{\Sigma}_N\mathbf{w},
%\end{equation}
%where 
%Using the Lagrange method, we can have a closed form solution given by 
%\begin{equation}
%\mathbf{w}_0 = \lambda \mathbf{\Sigma}^{-1}_N \mathbf{s}_0,
%\end{equation}
%where $\lambda$ is the Lagrange multiplier satisfying $\lambda = \frac{1}{\mathbf{s}_0^* \mathbf{\Sigma}^{-1}_N \mathbf{s}_0}$. However in most practical scenarios, the covariance matrix is unknown and thus has to be estimated. 
Using the RTE for covariance matrix estimation, the optimal MVDR beamforming vector becomes
\begin{equation}
\hat{\mathbf{u}}_0 = \frac{1}{\mathbf{s}_0^* \widehat{\mathbf{C}}_N^{-1} (\rho) \mathbf{s}_0}\widehat{\mathbf{C}}_N^{-1} (\rho) \mathbf{s}_0
\end{equation}
Therefore, the SNR at the output of the MVDR beamforming is given by:
\begin{equation}
\label{snr_mvdr}
\widehat{\textnormal{SNR}}\left(\rho\right)=\frac{\left(\mathbf{s^*_0}\widehat{\mathbf{C}}_N^{-1} (\rho)\mathbf{s_0}\right)^2}{\mathbf{s^*_0}\widehat{\mathbf{C}}_N^{-1} (\rho)\mathbf{\Sigma}_N\widehat{\mathbf{C}}_N^{-1} (\rho)\mathbf{s_0}}.
\end{equation}
\section{Asymptotic Behaviour of the MVDR Beamforming SNR}
In this paper, our aim is to study the first and second-order statistics of the SINR in \eqref{snr_mvdr}. For the sake of tractability, this study is carried out under two asymptotic regimes. The first one corresponds to $N$ and $n$ growing to infinity such that $c_N\triangleq \frac{N}{n}\to c$ and is referred to as the large-$(N,n)$ regime, whereas the second one considers the case of fixed $N$  with $n$ growing to infinity and will be coined the Large$-n$ regime. 
\subsection{Asymptotic Behavior in the Large-$(N,n)$ Regime}
In this section, we study the fluctuations of the SINR in the large$(N,n)$-regime. To this end, we will essentially rely on the second order analysis of the SNR at the output of the MVDR established in \cite{rubio} and  the recent results concerning the behaviour of quadratic forms associated with the RTE \cite{abla_romain}. Details of the derivation are provided in Appendix 1. 
%In this part, we consider the random matrix theory (RMT) regime where $N,n \rightarrow \infty$ and $c_N\triangleq \frac{N}{n} \rightarrow c \in \left(0,\infty\right)$. Based on results on the second order statistics of bilinear forms of robust scatter estimators derived in \cite{abla_romain} and the second order analysis done in \cite{rubio}, we show that a central limit theorem on the MVDR SNR can be derived in the large-$(N,n)$ regime. This constitutes the findings of the following theorem. 
Before stating our first main result, we will introduce some notations (see \cite{abla_romain} and \cite{rubio}).
%Before stating our theorem,  we shall introduce some useful notations. \\
We define $\gamma_N\left(\rho\right)$ to be  the solution to the following equation:
\begin{align*}
    1 = \frac{1}{N}\text{tr}\left[\boldsymbol{\Sigma}_N\left(\rho\gamma_N\left(\rho\right)\mathbf{I}_N+\left(1-\rho\right)\boldsymbol{\Sigma}_N \right)^{-1}\right].
\end{align*}
We also denote by $\delta$ the solution to the following fixed-point equation
\begin{align*}
    \delta = \frac{1}{n} \text{tr}\left[\boldsymbol{\Sigma}_N\left(\frac{1}{1+\delta}\boldsymbol{\Sigma}_N+\alpha\left(\rho\right)\mathbf{I}_N\right)^{-1}\right],
\end{align*}
where 
\begin{align*}
\alpha\left(\rho\right)=\frac{\rho \gamma_N\left(\rho\right)\left(1-(1-\rho)c_N\right)}{1-\rho}.
\end{align*}
\begin{align*}
\overline{\textnormal{SNR}\left(\rho\right)}= \overline{\text{SNR}\left(\hat{\mathbf{w}}_{0,\text{MVDR}}\right)}, \quad \sigma_{N,n}^2 = \sigma_{s,M}^2.
\end{align*}
$\overline{\text{SNR}\left(\hat{\mathbf{w}}_{0,\text{MVDR}}\right)}$ and $\sigma_{s,M}^2$ are defined in \cite[Theorem 1]{rubio} by replacing $\delta_M$ by $\delta$, $\tilde{\delta}_M$ by $\frac{1}{1+\delta}$ and $\alpha$ by $\alpha\left(\rho\right)$.

\begin{theorem}
    Assume that $\boldsymbol{\Sigma}_N$ is given by \eqref{eq:sigma} where $q$ is fixed. In the Large $\left(N,n\right)$- regime where $(N,n)\rightarrow \infty$ with $\frac{N}{n} \rightarrow c$, the quantity $\sigma_{N,n}^{-1}\sqrt{n}\left(\widehat{\textnormal{SNR}}\left(\rho\right)-\overline{\textnormal{SNR}\left(\rho\right)}\right)$ behaves as a standard normal distribution or equivalently
\begin{equation}
\sigma_{N,n}^{-1}\sqrt{n}\left(\widehat{\textnormal{SNR}}\left(\rho\right)-\overline{\textnormal{SNR}\left(\rho\right)}\right)\xrightarrow[(N,n)\to+\infty]{d}\mathcal{N}\left(0,1\right).
\end{equation}
\end{theorem}
\begin{proof}
See Appendix A for a detailed proof.
\end{proof}
\subsection{Asymptotic Behavior in the large-$n$ Regime}
In this section, we study the fluctuations of the SNR at the output of the MVDR \eqref{snr_mvdr} in the large-$n$ regime.  
%In this section, we establish the convergence of the SNR in \eqref{snr_mvdr} to a Gaussian random variable in the classical regime when $n$ tends to infinity while $N$ is fixed, a regime which we refer to as the large$-n$ regime.  
Our result will mainly build on the CLT of the RTE that has recently been derived in \cite{abla}.
%This result is based on the asymptotic convergence of the robust scatter matrix that has recently been proved in \cite{abla}
%This result will prepare the ground for the computation of the outage probability analysis that will be addressed in the next section. 
%The asymptotic analysis of the SNR is based on the following recent result \cite{abla} proving the convergence and fluctuations of the RTE and which we have to 
 %It is worth mentioning that this result is reminiscent of that of \cite{rubio} dealing with the fluctuations of the SNR when a Gaussian noise is considered. 
%In this section, we provide our main theoretical result, where the goal is to establish the convergence and the fluctuations of the MVDR SNR in the regime $N$ fixed, $n$ tends to $\infty$ which we call large-$n$ regime for now on. Let us first recall an important result on the convergence and fluctuations of the RTE in the large-$n$ regime. More precisely, the authors in \cite{abla} established a CLT for the RTE.
 Keeping the same notations as in \cite{abla}, the following theorem from \cite{abla} establishes the CLT of the robust-scatter estimator:
\begin{lemma} \cite{abla}
In the large-$n$ regime, 
\begin{align*}
\sqrt{n}\left(\textnormal{vec}\left(\widehat{\mathbf{C}}_N (\rho)\right)-\textnormal{vec}\left(\mathbf{\Sigma}_0(\rho)\right)\right)
\end{align*} behaves as a zero-mean Gaussian distributed vector with covariance matrix $\mathbf{M}_1$ and pseudo-covariance matrix $\mathbf{M}_2$ defined in \cite{abla}, where $\mathbf{\Sigma}_0(\rho)$ is the solution to the following equation
\begin{equation}
\label{sigma0}
\mathbf{\Sigma}_0(\rho) = N\left(1-\rho\right)\mathbb{E}\left[\frac{\mathbf{x}\mathbf{x}^*}{\mathbf{x}^*\boldsymbol{\Sigma}_0^{-1}(\rho) \mathbf{x}}\right]+\rho \mathbf{I}_N,
\end{equation}
where the expectation is taken over the distribution of the random vectors $\mathbf{x}_i$. \footnote{
    A simple way to evaluate numerically $\boldsymbol{\Sigma}_0$ has been provided in \cite{abla}. It is merely based on noticing that the eigenvectors of $\boldsymbol{\Sigma}_0$ are the same as $\boldsymbol{\Sigma}_N$ while its eigenvalues  satisfy a fixed point equation as shown in \cite{abla}, %which provides a simple way to evaluate numerically $\boldsymbol{\Sigma}_0$. %can be evaluated numerically as shown  \cite{abla}. %in allowing the numerical evaluation of $\boldsymbol{\Sigma}_0(\rho)$ has been provided in \cite{abla}.
}
\end{lemma}
The following theorem can be used in order to derive CLT for any functional of the RTE under the large$-n$ regime. In particular, we will show in this work how this CLT can be transferred to that of the SNR at the output of the MVDR beamforming. 
Note that under large$-n,N$ regime  a similar result  cannot be derived in general as the dimensions of $\widehat{\bf C}_N(\rho)$ increase with the number of samples. Before stating our second main theorem, we shall introduce the following quantities:  
%This is the reason why for each functional deriva 
%Using advanced tools from probability and statistics, it is possible to transfer the CLT theorem  of the RTE to any functional of the RTE and thus to the SNR at the MVDR beamforming. This constitutes our principal result. Before stating our main theorem, we shall introduce some useful quantities:
%Before stating the main result of the paper, we shall some useful quantities
\small
\begin{align*}
&\mathbf{B}=\mathbf{\Sigma}_0^{-1}(\rho) \mathbf{\Sigma}_N\mathbf{\Sigma}_0^{-1}(\rho). \\
&\textnormal{SNR}_0\left(\rho\right)=\frac{\left(\mathbf{s^*_0}\mathbf{\Sigma}_0^{-1}(\rho)\mathbf{s_0}\right)^2}{\mathbf{s^*_0}\mathbf{B}\mathbf{s_0}}. \\ 
&\Xi= \frac{1}{2}\begin{bmatrix}
\Re(\mathbf{M}_1)+\Re(\mathbf{M}_2) & -\Im(\mathbf{M}_1)+\Im(\mathbf{M}_2)\\ 
 \Im(\mathbf{M}_1)+\Im(\mathbf{M}_2)& \Re(\mathbf{M}_1)-\Re(\mathbf{M}_2) 
\end{bmatrix} ,\\
&\mathbf{c}^* = \frac{\left(\mathbf{s}_0^*\mathbf{\Sigma}_0^{-1}(\rho)\mathbf{s_0}\right)^2}{\left(\mathbf{s}_0^*\mathbf{B}\mathbf{s_0}\right)^2} \biggl[ \mathbf{s}_0^*\mathbf{B}\left[\left(\mathbf{\Sigma}_0^{-1}(\rho)\mathbf{s_0}\right)^t \otimes \mathbf{I}_N\right]\\
&+ \mathbf{s}_0^*\mathbf{\Sigma}_0^{-1}(\rho)\left[\left(\mathbf{B}\mathbf{s_0}\right)^t \otimes \mathbf{I}_N\right]\Biggr]\\
&-2 \frac{\mathbf{s}_0^*\mathbf{\Sigma}_0^{-1}(\rho)\mathbf{s_0}}{\mathbf{s}_0^*\mathbf{B}\mathbf{s_0}} \mathbf{s}_0^*\mathbf{\Sigma}_0^{-1}(\rho)\left[\left(\mathbf{\Sigma}_0^{-1}(\rho)\mathbf{s_0}\right)^t \otimes \mathbf{I}_N\right]. \\
&\widetilde{\mathbf{c}} = \begin{bmatrix}
\Re\left(c \right )\\ 
\Im\left(c \right )
\end{bmatrix}.
\end{align*}
\normalsize
\begin{theorem}
In the large$-n$ regime $\sqrt{n}\left(\widehat{\textnormal{SNR}}\left(\rho\right)-\textnormal{SNR}_0\left(\rho\right)\right)$ behaves as a Normal distribution with zero-mean and variance $\sigma_n^2 =\widetilde{\mathbf{c}}^t \Xi \widetilde{\mathbf{c}} $ or equivalently
\begin{equation} \label{clt}
\sigma_n^{-1}\sqrt{n}\left(\widehat{\textnormal{SNR}}\left(\rho\right)-\textnormal{SNR}_0\left(\rho\right)\right)\xrightarrow[n\to+\infty]{d}\mathcal{N}\left(0,1\right).
\end{equation}
\end{theorem}
\begin{proof} See Appendix B for a detailed proof.
\end{proof}
\section{Numerical Results}
In all our simulations, we consider a uniform linear array (ULA) with elements located half a wavelength apart. The desired signal is received at an exploration angle $\theta_0=0$ deg, and the interfering signals are received from the angles $-35$ and $70$ degrees. Moreover, all signals were received at a power $10$ dB above the background noise. In all simulations, we fix the number of antennas to $N=4$. Moreover, we assume that the number of observations $n$ can not exceed $100$
observations ($\frac{n}{N}< 25$), which constitutes the total budget of the system in terms of samples used to estimate the noise-plus-interference covariance matrix. This assumption is quite practical and has been considered in many papers in the literature \cite{mestre,rubio2,ESPAR}. To assess the accuracy of the derived CLTs in both regimes, we will use two different metrics, namely  the symmetrized divergence Kolmogorov–-Smirnov (KS) statistic and the $f-$ divergence denoted as $\mathcal{D}_f\left(P ||
Q\right)$. These two metrics are generally employed to quantify the difference between two continuous probability distributions with CDFs $P$ and $Q$ respectively. The KS statistic between $P$ and $Q$ is given by
\begin{equation}
\label{KStest}
D \triangleq \sup_x \left|P\left(x\right)-Q\left(x\right)\right|.
\end{equation}
while the $f-$divergence with respect to a convex function $f$ satisfying $f\left(1\right)=0$, is defined as follows
\begin{equation}
\label{fdivergence}
\mathcal{D}_f\left(P || Q\right)=\int f\left(\frac{p\left(x\right)}{q\left(x\right)}\right)q\left(x\right)dx,
\end{equation}
where $p$ and $q$ are respectively the corresponding PDFs of $P$ and $Q$. In Table \ref{tab:measure}, we summarize some selected instances of the functions $f$ that we use in the letter.\footnote{Note that the Kullback-Leibler divergence defined in Table I is not a distance.  We thus  use instead its modified version called the symmetrised divergence given by $\mathcal{D}_{KL}\left(P||Q\right)+\mathcal{D}_{KL}\left(Q||P\right)$.}
\begin{table}[H]
\centering
\begin{center}
\begin{tabular}{|c|c|}
  \hline
  Divergence & Corresponding $f\left(t\right)$ \\
  \hline
  Hellinger distance, $\mathcal{H}\left(P||Q\right)$ & $\left(\sqrt{t}-1\right)^2$ \\
  \hline 
  Total variation distance, $\mathcal{T}\left(P||Q\right)$ & $\frac{1}{2}\left|t-1\right|$  \\
  \hline
  Kullback-Leibler divergence, $\mathcal{D}_{KL}\left(P||Q\right)$ & $t\log t$ \\
  \hline
\end{tabular}
\caption{Selected instances of the $f-$divergence} 
\label{tab:measure}
\end{center}
\end{table}

To have a unified notation, we denote by $\mathcal{D}\left(p,q\right)$, the distance between $p$ and $q$, where the metric can be either the KS statistic or the $f$-divergence. 
%
% In this letter, we only focus on the Kullback-Leibler (KL) divergence, the Hellinger distance,  
%
%
%In fact, the $f-$divergence has many variants
%
%
%
%
%
% refer to the Hellinger distance \cite{hellinger} that is used to quantify the similarity between two probability distributions. Given two random variable probability densities $f$ and $g$, the squared Hellinger distance between $f$ and $g$ is given by the following expression
%\begin{equation}
%\label{Hellinger}
%\begin{split}
%H^2\left(f,g\right)& = \frac{1}{2}\int \left(\sqrt{f\left(x\right)}-\sqrt{g\left(x\right)}\right)^2dx \\
%& = 1-\int \sqrt{f\left(x\right)g\left(x\right)}dx.
%\end{split}
%\end{equation}
With these metrics at hand,  we compare the empirical cumulative distribution function (CDF) of the following quantities
\begin{align*}
& Q_{N,n}\triangleq \sigma_{N,n}^{-1}\sqrt{n}\left(\widehat{\textnormal{SNR}}\left(\rho\right)-\overline{\textnormal{SNR}\left(\rho\right)}\right). \\
&Q_n \triangleq \sigma_n^{-1}\sqrt{n}\left(\widehat{\textnormal{SNR}}\left(\rho\right)-\textnormal{SNR}_0\left(\rho\right)\right).
\end{align*}
with that of the standard normal distribution $\mathcal{N}\left(0,1\right)$.   Letting $\overline{f}_{N,n}, \overline{F}_{N,n}$ and $\overline{f}_n, \overline{F}_{n}$  be respectively the empirical PDFs and CDFs of $Q_{N,n}$ and $Q_n$, we define the following distance metrics
\begin{equation}
\label{kullback_metric}
\begin{split}
&\mathcal{D}_{N,n} \triangleq \mathcal{D}\left(f_{0,1},\overline{f}_{N,n}\right). \\
&\mathcal{D}_{n} \triangleq \mathcal{D}\left(f_{0,1},\overline{f}_{n}\right).
\end{split}
\end{equation}
where $f_{0,1}$ and $F_{0,1}$ denote respectively the PDF and CDF of the standard normal distribution \footnote{$f_{0,1}\left(x\right)=\frac{1}{\sqrt{2\pi}}\exp\left(-\frac{x^2}{2}\right)$ and $F_{0,1}\left(x\right)=1-Q\left(x\right)$, where $Q\left(.\right)$ is the Q-function.}.
\begin{figure}[t!]
       \centering
    \includegraphics[width=3.5in]{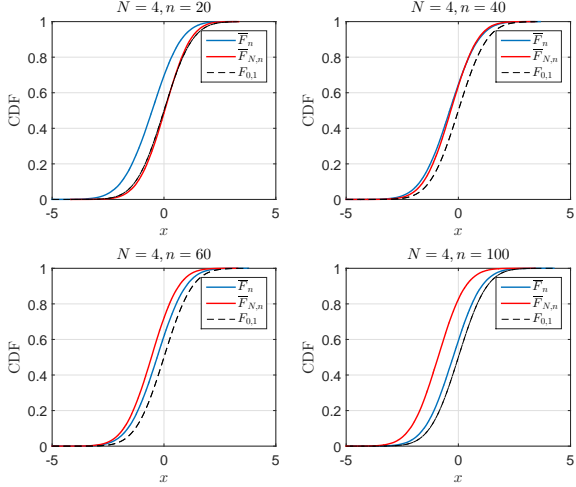}
\caption{Empirical cumulative distribution functions (CDF) of $Q_n$ and $Q_{N,n}$ compared with the CDF of the standard Normal density for different values of $\left(N,n\right)$. $\rho=0.65$.}
\label{fig:CLT_fig}
\end{figure}
%\begin{figure}[t!]
%       \centering
%    \includegraphics[width=3.5in]{Figures_letter/distance_rho_2.eps}
%\caption{Hellinger distance, $\mathcal{D}_{n}$ and $\mathcal{D}_{N,n}$ as defined in (\ref{kullback_metric}) versus $\frac{n}{N}$. $\rho=0.2$.}
%\label{fig:H_mid}
%\end{figure}
\begin{figure}[t!]
       \centering
    \includegraphics[width=3.5in]{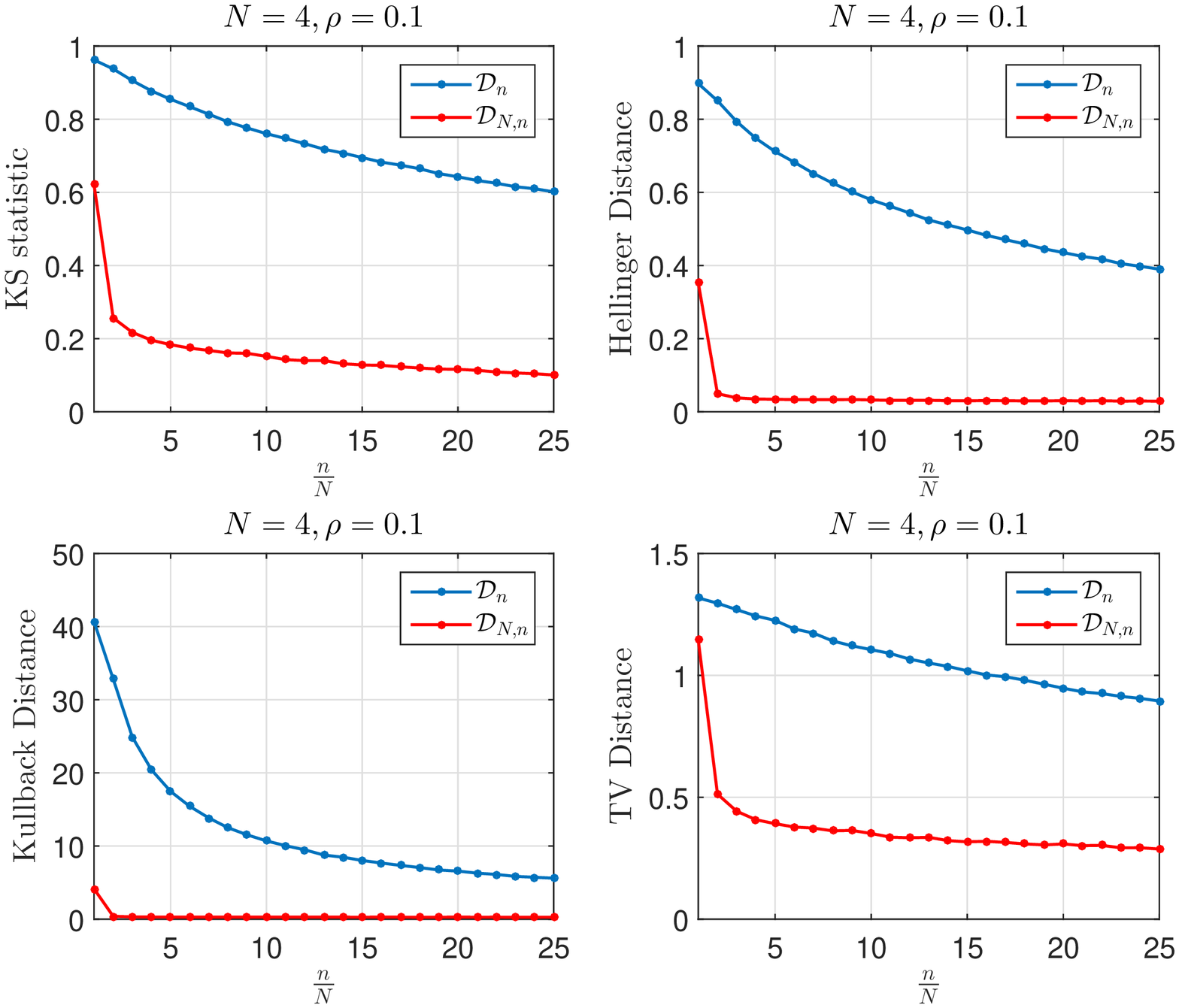}
\caption{$\mathcal{D}_{n}$ and $\mathcal{D}_{N,n}$ as defined in (\ref{kullback_metric}) versus $\frac{n}{N}$. $\rho=0.1$.}
\label{fig:H_small}
\end{figure}
\begin{figure}[t!]
       \centering
    \includegraphics[width=3.5in]{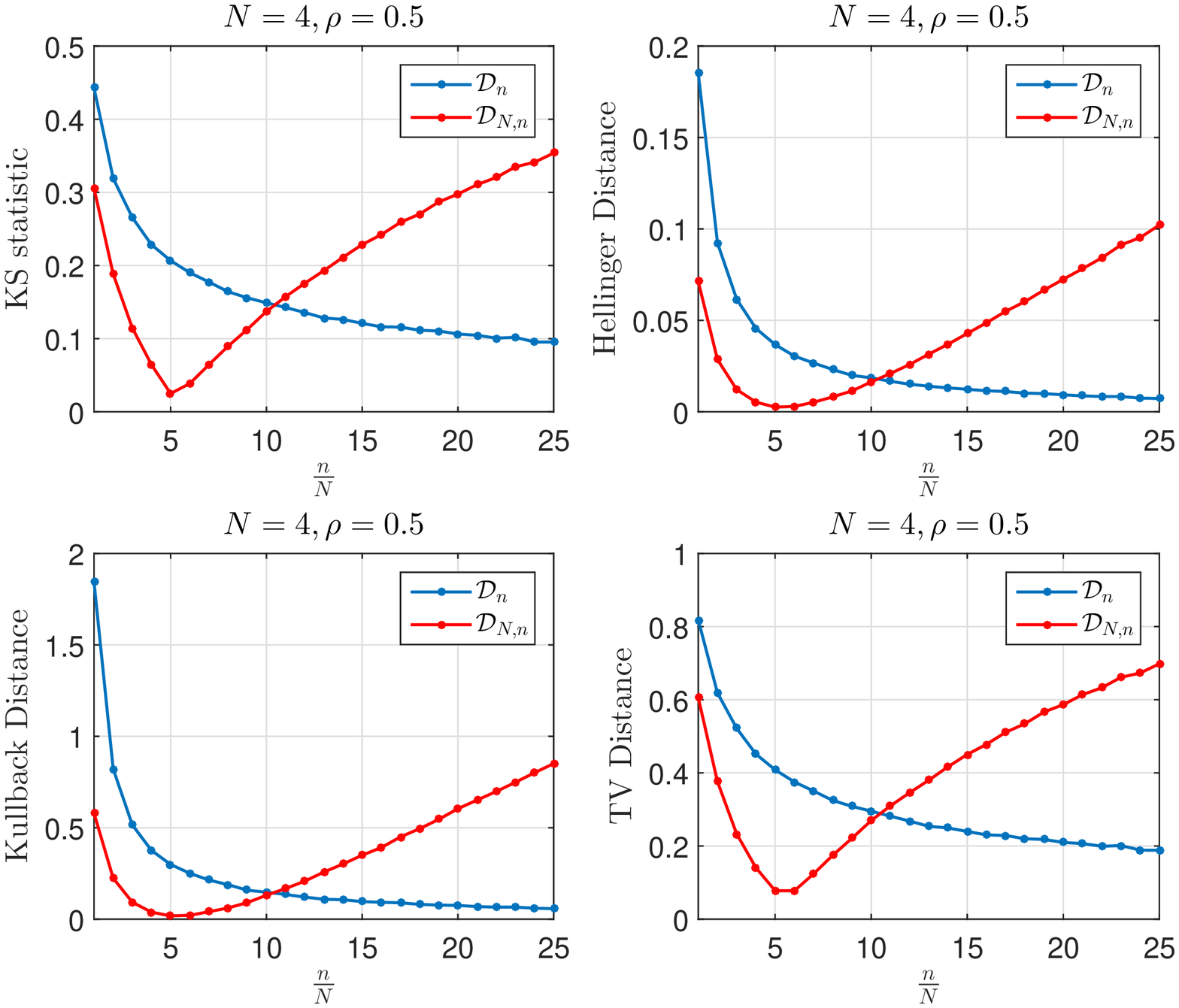}
\caption{$\mathcal{D}_{n}$ and $\mathcal{D}_{N,n}$ as defined in (\ref{kullback_metric}) versus $\frac{n}{N}$. $\rho=0.5$.}
\label{fig:H_mid}
\end{figure}
\begin{figure}[t!]
       \centering
    \includegraphics[width=3.5in]{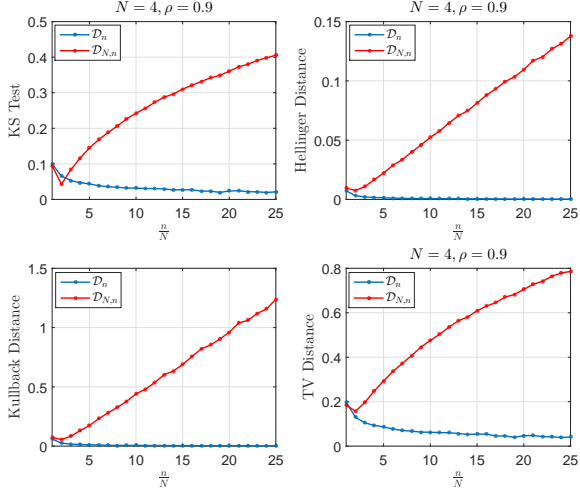}
\caption{$\mathcal{D}_{n}$ and $\mathcal{D}_{N,n}$ as defined in (\ref{kullback_metric}) versus $\frac{n}{N}$. $\rho=0.9$.}
\label{fig:H_high}
\end{figure}

To begin with, we report in Figure \ref{fig:CLT_fig} the empirical CDFs of  both  regimes along with the standard normal CDF. As seen, the accuracy the $n$-large regime is lower than the $\left(N,n\right)$-large regime for small values of $n$. As the number of observations increases, the $\left(N,n\right)-$large regime becomes less accurate. This is to be compared with the $n-$large regime which provides a good fit
%As the number of observations $n$ increases, the $\left(N,n\right)-$large regime seems to diverge from the standard normal distribution, which is not the case for the $n-$large regime wherein  a good
 for $n=60$ and $n=100$.\par 
 To reinforce the observations made in Figure \ref{fig:CLT_fig}, we display the distance between the empirical distributions $\overline{f}_n$ and $\overline{f}_{N,n}$ with the standard normal density for different values of $N$ and $n$. In order to investigate the impact of $\rho$ on the performance, we display the different instances of the distance for small ($\rho=0.1$  Figure \ref{fig:H_small}), mid ($\rho=0.5$, Figure \ref{fig:H_mid}) and high ($\rho=0.9$, Figure \ref{fig:H_high}) values of $\rho$. 
 \begin{itemize}
     \item $\rho=0.1$: We observe in this case that the $(N,n)$-Large regime provides more accurate results. This might be related to the fact that the empirical average $  \frac{1}{n}\sum_{i=1}^n \frac{\mathbf{x}_i\mathbf{x}_i^*}{\frac{1}{N}\mathbf{x}_i^*\widehat{\mathbf{C}}_N^{-1} (\rho)\mathbf{x}_i}$ is the dominant term in the expression of $\widehat{\bf C}_N(\rho)$. This averaging is approximated in the $n$-Large regime by
         $N\left(1-\rho\right)\mathbb{E}\left[\frac{\mathbf{x}\mathbf{x}^*}{\mathbf{x}^*\mathbf{\Sigma}_0^{-1}(\rho) \mathbf{x}}\right]$, which can be not well-estimated as $n$ is limited to $100$.  On the other hand, the $(n,N)-$ large regime is more accurate since it leverages the double averaging over $n$ and $N$. 
     %In this case, $\mathbf{\Sigma}_0(\rho)\simeq N\left(1-\rho\right)\mathbb{E}\left[\frac{\mathbf{x}\mathbf{x}^*}{\mathbf{x}^*\mathbf{\Sigma}_0^{-1}(\rho) \mathbf{x}}\right]$ which \textcolor{red}{might} be far from the actual estimate $\widehat{\mathbf{C}}_N (\rho) \simeq \frac{1}{n}\sum_{i=1}^n \frac{\mathbf{x}_i\mathbf{x}_i^*}{\frac{1}{N}\mathbf{x}_i^*\widehat{\mathbf{C}}_N^{-1} (\rho)\mathbf{x}_i}$ that consists of simple averaging over a finite number of observations. This explains why the $n-$Large regime fails to approximate the mean \text{SNR} and thus diverge from the standard normal distribution for finite number of observations. In counterpart, the estimate given by the $\left(N,n\right)-$large regime seems to be more robust to the finite sample size impairments. This is mainly due to the fact that the $\left(N,n\right)-$large regime assumes that both dimensions grows simultaneously at the same pace which is of practical interest in this scenario.
     \item $\rho=0.9$: In this case, $\rho$ is high, thus, the estimated \text{SNR} behaves as a deterministic quantity since $\widehat{\mathbf{C}}_N (\rho) \simeq \mathbf{I}_N$ which is the case for $\mathbf{\Sigma}_0(\rho)$ as well. This is clearly expected by the large$-n$ regime. However, the large$-(n,N)$ regime fails to predict in an accurate way the performances. One possible explanation can be related to the fact that as $\rho$ tends to $1$, quantity $\alpha(\rho)$ converges to
         infinity, causing the fluctuations to be not properly predicted.  
     %As a matter of fact, the $n-$large regime is expected to approach the actual \text{SNR} with small error that vanishes as $\rho \rightarrow 1$. However, the $\left(N,n\right)-$large regime fails to estimate the true SNR, that's why the empirical distribution is shifted away from zero which deteriorates the divergence metrics\footnote{Note that both regimes perfectly approach the normal distribution for $\rho=1$, but for values of $\rho$ less than $1$, they have different behaviors.}.
 \item $\rho=0.5:$ In this case, the main observation is that the $\left(N,n\right)-$large regime has a better fit to the standard normal distribution for lower values of $n$, while for $\frac{n}{N}\geq10$, the $n-$large regime starts to exhibit a better fit.
 \end{itemize}
 As a conclusion, for mid values of $\rho$, it is better to work under the $\left(N,n\right)-$large regime as long as the number of observations is low. As we get more observations, the $n-$large regime yields a better performance.
\section{Conclusion}
In this paper, we have analyzed the asymptotic behaviour of the Capon's MVDR beamformer when using the regularized Tyler estimator (RTE) for both the $n-$large and the $\left(N,n\right)-$large regimes. Based on recent results on the convergence of the RTE, we have analyzed the fluctuations of the SNR at the output of the MVDR. Using well known divergence metrics, we have examined the accuracy of both regimes and determined which regime is more accurate and thus more convenient to use.
%In this paper, we studied the performance of the Capon's MVDR beamformer where we considered the regularized Tyler estimator for covariance matrix estimation. To derive the system outage probability, we rely on asymptotic analysis where we derive a CLT on the MVDR SNR. The obtained result is a function of the RTE shrinkage parameter. By means of this result, we obtain the optimal $\rho$ that minimize the outage probability.
% Numerical results were provided in order to validate our theoretical findings and to show the superiority of the optimized RTE. 
 \section*{Appendix A}
\section*{Proof of Theorem 1}
The proof hinges on recent  results concerning the asymptotic behaviour of the RTE developed in \cite{abla_romain} and the second order analysis of the SNR at the output of the MVDR derived  in \cite{rubio}. As we shall see next, these results  lead together  to the sought-for CLT. 
%The proof is mainly inspired by the second order analysis of the MVDR' SNR established in \cite{rubio} as well as the matrix approximation proposed in \cite{abla_romain}. As a matter of fact, the proof is mainly to show how to establish the connection between the two results. In the sequel, we provide the two elements of the proof and show how can we connect them to obtain the claim of Theorem 2. 
\subsubsection{First ingredient: Approximation of the RTE estimate $\widehat{\mathbf{C}}_N (\rho)$}
Studying the RTE in the large$-n,N$ regime is not an easy task, as the RTE does not follow a standard random matrix model. To overcome these issues, the work in \cite{abla_romain} shows that as far as quadratic forms are concerned the asymptotic behaviour of the $\widehat{\mathbf{C}}_N (\rho)$ is the same as another random object which, contrary to $\widehat{\mathbf{C}}_N (\rho)$ can be studied using standard RMT tools. More formally, we have the following convergence
results,  
%In \cite{abla_romain}, the authors proposed to approximate the RTE estimate $\widehat{\mathbf{C}}_N (\rho)$ by a matrix $\widehat{\mathbf{S}}_N (\rho)$ that follows a classical random matrix model. The approximation is in the sense that for any $k \in \mathbb{Z}$ and $\epsilon > 0$, we have
\begin{equation}
\label{convergence}
N^{1-\epsilon}|\mathbf{u}^*\widehat{\mathbf{C}}^k_N (\rho)\mathbf{v}-\mathbf{u}^*\widehat{\mathbf{S}}^k_N (\rho)\mathbf{v} | \xrightarrow[]{a.s.}0,
\end{equation}
%where $\mathcal{R}_{\kappa} \triangleq \left[\kappa+\max\left(0,1-c^{-1},1\right)\right]$, $\kappa$ small
where $\mathbf{u}$ and $\mathbf{v}$ are unit norm vectors in $\mathbb{C}^N$ and $\widehat{\bf S}_N(\rho)$ is given by%Moreover, the approximation matrix is given by
\begin{align*}
\widehat{\mathbf{S}}_N (\rho) = \frac{\rho}{\alpha\left(\rho\right)} \frac{1}{n}\sum_{i=1}^n \mathbf{z}_i \mathbf{z}_i^* + \rho \mathbf{I}_N, \mathbf{z}_i= \mathbf{\Sigma}_N^{\frac{1}{2}} \mathbf{w}_i, \: i=1,\cdots,n.
\end{align*}
As will be shown next, this convergence implies that the SNR has the same fluctuations if $\widehat{\bf C}_N$ is replaced by $\widehat{\mathbf{S}}_N (\rho)$. 

%This statement will be proved rigorously  in the second part of this appendix.  
As the SNR is scale invariant, the fluctuations would be the same when  $\widehat{\bf C}_N$ is replaced by $\widetilde{\bf S}_N(\rho)$ given by:
\begin{equation}
 \widetilde{\bf S}_N(\rho)=\frac{1}{n}\sum_{i=1}^n \mathbf{z}_i \mathbf{z}_i^* + \alpha\left(\rho\right)\mathbf{I}_N.
 \label{approximation}
 \end{equation}

%As it is clear from the MVDR' SNR expression in (\ref{snr_mvdr}), scaling the estimate of the covariance matrix does not affect the MVDR performance in terms of SNR. Thus, scaling $\widehat{\mathbf{S}}_N (\rho)$ by $\frac{\alpha\left(\rho\right)}{\rho}$ yields the same SNR as $\widehat{\mathbf{S}}_N (\rho)$. Therefore, without loss of generality, $\widehat{\mathbf{S}}_N (\rho)$ adopts the following expression
%\begin{equation}
%\label{approximation} 
%\widehat{\mathbf{S}}_N (\rho) =  \frac{1}{n}\sum_{i=1}^n \mathbf{z}_i \mathbf{z}_i^* + \alpha\left(\rho\right)\mathbf{I}_N.
%\end{equation}
 Interestingly, the approximation matrix in (\ref{approximation}) follows the same structure of a sample correlation matrix with a diagonal loading factor $\alpha\left(\rho\right)$ that was considered in \cite{mestre}. The fluctuations of the SNR will be thus obtained by simply leveraging the results of \cite{mestre}. 
%Establishing the CLT for that kind of matrices %is a simpler task which has been done in \cite{mestre,rubio}.
\subsubsection{Second-order Analysis of the SNR of Diagonally Loaded MVDR Filters}
As discussed above, to prove Theorem 1, it suffices to show that the SNR has the same fluctuations when the RTE is replaced by $\widehat{\bf S}_N(\rho)$. 
%In \cite{rubio}, the authors established a CLT on the SNR of diagonally loaded MVDR filters in the presence of both noise spatial and temporal correlation matrices. In our scenario, the samples are uncorrelated in time. Thus, we can exploit the result by taking the temporal correlation matrix to be identity. A direct application of \cite[Theorem 1]{rubio}, we get the following result
%\begin{equation}
%\small
%\label{clt_rmt}
%\sigma_{N,n}^{-1}\sqrt{n}\left(\frac{\left(\mathbf{s^*_0}\widehat{\mathbf{S}}_N^{-1} (\rho)\mathbf{s_0}\right)^2}{\mathbf{s^*_0}\widehat{\mathbf{S}}_N^{-1} (\rho)\mathbf{\Sigma}_N\widehat{\mathbf{S}}_N^{-1} (\rho)\mathbf{s_0}}-\overline{\textnormal{SNR}\left(\rho\right)}\right) \xrightarrow[\frac{N}{n}\to c]{d}\mathcal{N}\left(0,1\right).
%\end{equation}
%Therefore, to get the claim of Theorem 2, it is sufficient to prove that 
As per the Slutsky Lemma, this amounts to showing that:
\begin{align*}
    &\sqrt{n}\left(\frac{\left(\mathbf{s^*_0}\widehat{\mathbf{C}}_N^{-1} (\rho)\mathbf{s_0}\right)^2}{\mathbf{s^*_0}\widehat{\mathbf{C}}_N^{-1} (\rho)\mathbf{\Sigma}_N\widehat{\mathbf{C}}_N^{-1} (\rho)\mathbf{s_0}}-\frac{\left(\mathbf{s^*_0}\widehat{\mathbf{S}}_N^{-1} (\rho)\mathbf{s_0}\right)^2}{\mathbf{s^*_0}\widehat{\mathbf{S}}_N^{-1} (\rho)\mathbf{\Sigma}_N\widehat{\mathbf{S}}_N^{-1} (\rho)\mathbf{s_0}}\right)\\
&
\xrightarrow[\frac{N}{n}\to c]{a.s.}0.
\end{align*}
\normalsize
To this end, we first decompose the above term as: 
\begin{align*}
&\sqrt{n}\left(\frac{\left(\mathbf{s^*_0}\widehat{\mathbf{C}}_N^{-1} (\rho)\mathbf{s_0}\right)^2}{\mathbf{s^*_0}\widehat{\mathbf{C}}_N^{-1} (\rho)\mathbf{\Sigma}_N\widehat{\mathbf{C}}_N^{-1} (\rho)\mathbf{s_0}}-\frac{\left(\mathbf{s^*_0}\widehat{\mathbf{S}}_N^{-1} (\rho)\mathbf{s_0}\right)^2}{\mathbf{s^*_0}\widehat{\mathbf{S}}_N^{-1} (\rho)\mathbf{\Sigma}_N\widehat{\mathbf{S}}_N^{-1} (\rho)\mathbf{s_0}}\right) \\
& = \sqrt{n}\left(\frac{\left(\mathbf{s^*_0}\widehat{\mathbf{C}}_N^{-1} (\rho)\mathbf{s_0}\right)^2}{\mathbf{s^*_0}\widehat{\mathbf{C}}_N^{-1} (\rho)\mathbf{\Sigma}_N\widehat{\mathbf{C}}_N^{-1} (\rho)\mathbf{s_0}}-\frac{\left(\mathbf{s^*_0}\widehat{\mathbf{S}}_N^{-1} (\rho)\mathbf{s_0}\right)^2}{\mathbf{s^*_0}\widehat{\mathbf{C}}_N^{-1} (\rho)\mathbf{\Sigma}_N\widehat{\mathbf{C}}_N^{-1} (\rho)\mathbf{s_0}}\right) \\
& + \sqrt{n}\left(\frac{\left(\mathbf{s^*_0}\widehat{\mathbf{S}}_N^{-1} (\rho)\mathbf{s_0}\right)^2}{\mathbf{s^*_0}\widehat{\mathbf{C}}_N^{-1} (\rho)\mathbf{\Sigma}_N\widehat{\mathbf{C}}_N^{-1} (\rho)\mathbf{s_0}}-\frac{\left(\mathbf{s^*_0}\widehat{\mathbf{S}}_N^{-1} (\rho)\mathbf{s_0}\right)^2}{\mathbf{s^*_0}\widehat{\mathbf{S}}_N^{-1} (\rho)\mathbf{\Sigma}_N\widehat{\mathbf{S}}_N^{-1} (\rho)\mathbf{s_0}}\right) \\
& \triangleq \xi_1 + \xi_2.
\end{align*}
The term $\xi_1$ can be rewritten as 
\begin{equation}
\label{xi1}
\begin{split}
\xi_1 &= \frac{\sqrt{n}\left(\mathbf{s^*_0}\widehat{\mathbf{C}}_N^{-1} (\rho)\mathbf{s_0}-\mathbf{s^*_0}\widehat{\mathbf{S}}_N^{-1} (\rho)\mathbf{s_0}\right)}{\mathbf{s^*_0}\widehat{\mathbf{C}}_N^{-1} (\rho)\mathbf{\Sigma}_N\widehat{\mathbf{C}}_N^{-1} (\rho)\mathbf{s_0}}. \\
& \times \left(\mathbf{s^*_0}\widehat{\mathbf{C}}_N^{-1} (\rho)\mathbf{s_0}+\mathbf{s^*_0}\widehat{\mathbf{S}}_N^{-1} (\rho)\mathbf{s_0}\right)
\end{split}
\end{equation}
Then, by the results of (\ref{convergence}), we have the following convergence
\begin{align*}
&\sqrt{n}\left(\mathbf{s^*_0}\widehat{\mathbf{C}}_N^{-1} (\rho)\mathbf{s_0}-\mathbf{s^*_0}\widehat{\mathbf{S}}_N^{-1} (\rho)\mathbf{s_0}\right) \xrightarrow[\frac{N}{n}\to c]{a.s.}0 \\
%&\mathbf{s^*_0}\widehat{\mathbf{C}}_N^{-1} (\rho)\mathbf{s_0} \xrightarrow[\frac{N}{n}\to c]{a.s.} \mathbf{s^*_0}\widehat{\mathbf{S}}_N^{-1} (\rho)\mathbf{s_0}.
\end{align*}
Moreover, since,
\begin{align*}
\left \| \widehat{\mathbf{C}}_N^{-1} (\rho)-\widehat{\mathbf{S}}_N^{-1} (\rho) \right \|\xrightarrow[]{a.s.}0.
\end{align*}
then, any well-behaved functional of $\widehat{\mathbf{C}}_N^{-1} (\rho)$ converges $a.s.$ to the same functional of $\widehat{\mathbf{S}}_N^{-1} (\rho)$. In particular, we do have:
$$
\mathbf{s^*_0}\widehat{\mathbf{C}}_N^{-1} (\rho)\mathbf{s_0} - \mathbf{s^*_0}\widehat{\mathbf{S}}_N^{-1} (\rho)\mathbf{s_0}\xrightarrow[\frac{N}{n}\to c]{a.s.} 0.
$$
and
%the convergence holds for quadratic forms as in the denominator of $\xi_1$ in (\ref{xi1}), i.e.
\begin{align*}
\mathbf{s^*_0}\widehat{\mathbf{C}}_N^{-1} (\rho)\mathbf{\Sigma}_N\widehat{\mathbf{C}}_N^{-1} (\rho)\mathbf{s_0} - \mathbf{s^*_0}\widehat{\mathbf{S}}_N^{-1} (\rho)\mathbf{\Sigma}_N\widehat{\mathbf{S}}_N^{-1} (\rho)\mathbf{s_0}\xrightarrow[]{a.s.}0.
\end{align*}
%Since, $\widehat{\mathbf{S}}_N(\rho)$ follows a classical random matrix model \cite{silverstein}, the quantities $\mathbf{s^*_0}\widehat{\mathbf{S}}_N^{-1} (\rho)\mathbf{s_0}$ and $\mathbf{s^*_0}\widehat{\mathbf{S}}_N^{-1} (\rho)\mathbf{\Sigma}_N\widehat{\mathbf{S}}_N^{-1} (\rho)\mathbf{s_0}$ are bounded in probability which permits to conclude that. In fact, the matrix $\widehat{\mathbf{S}}_N(\rho)$ can be written in the resolvent form, $\widehat{\mathbf{S}}_N(\rho)=\frac{1}{n}\mathbf{Z}\mathbf{Z}^* + \alpha\left(\rho\right)\mathbf{I}_N$, where $\mathbf{Z}=\left[\mathbf{z}_1,\cdots,\mathbf{z}_n\right]$. Based on results from \cite{walid},
All this leads to
\begin{align*}
\xi_1 \xrightarrow[\frac{N}{n}\to c]{P} 0.
\end{align*}
%\begin{align*}
%&\mathbf{s^*_0}\widehat{\mathbf{C}}_N^{-1} (\rho)\mathbf{\Sigma}_N\widehat{\mathbf{C}}_N^{-1} (\rho)\mathbf{s_0} \\
%&= \sigma_0^2\left(\mathbf{s^*_0}\widehat{\mathbf{C}}_N^{-2} (\rho)\mathbf{s_0}\right) 
% + \sum_{i=1}^q \sigma_i^2\mathbf{s^*_0}\widehat{\mathbf{C}}_N^{-1} (\rho) \mathbf{a}\left(\theta_i\right)\mathbf{a}\left(\theta_i\right)^*\widehat{\mathbf{C}}_N^{-1} (\rho)\mathbf{s_0}.
%\end{align*}
We now handle the term $\xi_2$. By a similar reasoning, $\xi_2$ can be rewritten as follows
\begin{align*}
\xi_2 
& = \\
& \frac{\left(\mathbf{s^*_0}\widehat{\mathbf{S}}_N^{-1} (\rho)\mathbf{s_0}\right)^2} {\left(\mathbf{s^*_0}\widehat{\mathbf{C}}_N^{-1} (\rho)\mathbf{\Sigma}_N\widehat{\mathbf{C}}_N^{-1} (\rho)\mathbf{s_0}\right)\left(\mathbf{s^*_0}\widehat{\mathbf{S}}_N^{-1} (\rho)\mathbf{\Sigma}_N\widehat{\mathbf{S}}_N^{-1} (\rho)\mathbf{s_0}\right)} \\
& \times \sqrt{n} \left(\mathbf{s^*_0}\widehat{\mathbf{S}}_N^{-1} (\rho)\mathbf{\Sigma}_N\widehat{\mathbf{S}}_N^{-1} (\rho)\mathbf{s_0} - \mathbf{s^*_0}\widehat{\mathbf{C}}_N^{-1} (\rho)\mathbf{\Sigma}_N\widehat{\mathbf{C}}_N^{-1} (\rho)\mathbf{s_0}\right).
\end{align*}
We now refer to the special structure of $\mathbf{\Sigma}_N$ and rewrite $\xi_2$ as follows
\begin{align*}
\xi_2
& = \frac{\left(\mathbf{s^*_0}\widehat{\mathbf{S}}_N^{-1} (\rho)\mathbf{s_0}\right)^2} {\left(\mathbf{s^*_0}\widehat{\mathbf{C}}_N^{-1} (\rho)\mathbf{\Sigma}_N\widehat{\mathbf{C}}_N^{-1} (\rho)\mathbf{s_0}\right)\left(\mathbf{s^*_0}\widehat{\mathbf{S}}_N^{-1} (\rho)\mathbf{\Sigma}_N\widehat{\mathbf{S}}_N^{-1} (\rho)\mathbf{s_0}\right)} \\
& \times \sqrt{n} \Biggl[\sigma_0^2 \left(\mathbf{s^*_0}\widehat{\mathbf{S}}_N^{-2} (\rho)\mathbf{s_0}-\mathbf{s^*_0}\widehat{\mathbf{C}}_N^{-2} (\rho)\mathbf{s_0}\right) \\
& + \sum_{i=1}^q \sigma_i^2 \left(|\mathbf{s^*_0}\widehat{\mathbf{S}}_N^{-1} (\rho)\mathbf{a}\left(\theta_i\right)|^2-|\mathbf{s^*_0}\widehat{\mathbf{C}}_N^{-1} (\rho)\mathbf{a}\left(\theta_i\right)|^2\right)\Biggr].
\end{align*}
Noticing that 
\begin{align*}
&|\mathbf{s^*_0}\widehat{\mathbf{S}}_N^{-1} (\rho)\mathbf{a}\left(\theta_i\right)|^2-|\mathbf{s^*_0}\widehat{\mathbf{C}}_N^{-1} (\rho)\mathbf{a}\left(\theta_i\right)|^2  \\
& = \left(|\mathbf{s^*_0}\widehat{\mathbf{S}}_N^{-1} (\rho)\mathbf{a}\left(\theta_i\right)|-|\mathbf{s^*_0}\widehat{\mathbf{C}}_N^{-1} (\rho)\mathbf{a}\left(\theta_i\right)|\right) \\
& \times \left(|\mathbf{s^*_0}\widehat{\mathbf{S}}_N^{-1} (\rho)\mathbf{a}\left(\theta_i\right)|+|\mathbf{s^*_0}\widehat{\mathbf{C}}_N^{-1} (\rho)\mathbf{a}\left(\theta_i\right)|\right).
\end{align*}
and resorting to the same arguments   used in the control of  $\xi_1$, it follows that 
\begin{align*}
\xi_2 \xrightarrow[\frac{N}{n}\to c]{a.s.} 0.
\end{align*}
This concludes the proof of Theorem 1.
\section*{Appendix B}
\section*{Proof of Theorem 2}
For ease of presentation, we omit the argument $\rho$ in the SNR expressions.
%We start from the MVDR beamforming SNR expression 
%\begin{equation}
%\widehat{\textnormal{SNR}} = \frac{\left(\mathbf{s^*_0}\widehat{\mathbf{C}}^{-1}_N (\rho)\mathbf{s_0}\right)^2}{\mathbf{s^*_0}\widehat{\mathbf{C}}^{-1}_N (\rho)\mathbf{\Sigma}_N\widehat{\mathbf{C}}^{-1}_N (\rho)\mathbf{s_0}}.
%\end{equation}
According to \cite{abla}, the asymptotic limit of $\widehat{\textnormal{SNR}}$ would be
\begin{align*}
\textnormal{SNR}_0=\frac{\left(\mathbf{s^*_0}\mathbf{\Sigma}_0^{-1}(\rho)\mathbf{s_0}\right)^2}{\mathbf{s^*_0}\mathbf{\Sigma}_0^{-1}(\rho)\mathbf{\Sigma}_N\mathbf{\Sigma}_0^{-1}(\rho)\mathbf{s_0}}.
\end{align*}
The objective here is to study the fluctuations of the SNR around $\textnormal{SNR}_0$. To this end, we decompose $\sqrt{n}\left(\widehat{\textnormal{SNR}}-\textnormal{SNR}_0\right)$ by subtracting and adding $\sqrt{n}\frac{\left(\mathbf{s^*_0}\mathbf{\Sigma}_0^{-1}(\rho)\mathbf{s_0}\right)^2}{\mathbf{s^*_0}\widehat{\mathbf{C}}^{-1}_N (\rho)\mathbf{\Sigma}_N\widehat{\mathbf{C}}^{-1}_N (\rho)\mathbf{s_0}}$ resulting in expression \eqref{Q1Q2} given on the top of the next page.
%Thus to study the fluctuations of the MVDR beamforming SNR around its limit, we need to examine the behaviour of the quantity 
%$\sqrt{n}\left(\widehat{\textnormal{SNR}}-\textnormal{SNR}_0\right)$.
%\begin{align*}
%\sqrt{n}\left(\widehat{\textnormal{SNR}}-\textnormal{SNR}_0\right)& = \sqrt{n}\Biggl[\frac{\left(\mathbf{s^*_0}\widehat{\mathbf{C}}^{-1}_N (\rho)\mathbf{s_0}\right)^2}{\mathbf{s^*_0}\widehat{\mathbf{C}}^{-1}_N (\rho)\mathbf{\Sigma}_N\widehat{\mathbf{C}}^{-1}_N (\rho)\mathbf{s_0}}\\
%&-\frac{\left(\mathbf{s^*_0}\mathbf{\Sigma}_0^{-1}(\rho)\mathbf{s_0}\right)^2}{\mathbf{s^*_0}\mathbf{\Sigma}_0^{-1}(\rho)\mathbf{\Sigma}_N\mathbf{\Sigma}_0^{-1}(\rho)\mathbf{s_0}} \Biggr].
%\end{align*}
%Substracting and adding $\sqrt{n}\frac{\left(\mathbf{s^*_0}\mathbf{\Sigma}_0^{-1}(\rho)\mathbf{s_0}\right)^2}{\mathbf{s^*_0}\widehat{\mathbf{C}}^{-1}_N (\rho)\mathbf{\Sigma}_N\widehat{\mathbf{C}}^{-1}_N (\rho)\mathbf{s_0}}$ from the above expression, we get the expressions given by (\ref{Q1Q2}) which is displayed on the top of the next page.
\begin{figure*}[!t]
\begin{equation} \label{Q1Q2}
\begin{split}
&\sqrt{n}\left(\widehat{\textnormal{SNR}}-\textnormal{SNR}_0\right) =\\
& \underbrace{\frac{\sqrt{n}}{\mathbf{s^*_0}\widehat{\mathbf{C}}^{-1}_N (\rho)\mathbf{\Sigma}_N\widehat{\mathbf{C}}^{-1}_N (\rho)\mathbf{s_0}}\left(\mathbf{s^*_0}\widehat{\mathbf{C}}^{-1}_N (\rho)\mathbf{s_0}-\mathbf{s^*_0}\mathbf{\Sigma}_0^{-1}(\rho)\mathbf{s_0}\right)\left(\mathbf{s^*_0}\widehat{\mathbf{C}}^{-1}_N (\rho)\mathbf{s_0}+\mathbf{s^*_0}\mathbf{\Sigma}_0^{-1}(\rho)\mathbf{s_0}\right)}_{Q_1} \\
&+ \underbrace{\frac{\sqrt{n}\left(\mathbf{s^*_0}\mathbf{\Sigma}_0^{-1}(\rho)\mathbf{s_0}\right)^2}{\left(\mathbf{s^*_0}\widehat{\mathbf{C}}^{-1}_N (\rho)\mathbf{\Sigma}_N\widehat{\mathbf{C}}^{-1}_N (\rho)\mathbf{s_0}\right)\left(\mathbf{s^*_0}\mathbf{\Sigma}_0^{-1}(\rho)\mathbf{\Sigma}_N\mathbf{\Sigma}_0^{-1}(\rho)\mathbf{s_0}\right)}\left(\mathbf{s^*_0}\mathbf{\Sigma}_0^{-1}(\rho)\mathbf{\Sigma}_N\mathbf{\Sigma}_0^{-1}(\rho)\mathbf{s_0}-\mathbf{s^*_0}\widehat{\mathbf{C}}^{-1}_N (\rho)\mathbf{\Sigma}_N\widehat{\mathbf{C}}^{-1}_N (\rho)\mathbf{s_0}\right)}_{Q_2}
\end{split}.
\end{equation}
\hrulefill
\vspace*{4pt}
\end{figure*}
We will now treat subsequently the terms $Q_1$ and $Q_2$ defined in \eqref{Q1Q2}.
%We treat each term separately. We start with $Q_1$ defined in (\ref{Q1Q2}):
First, note that using the resolvent identity:
\begin{equation}
\widehat{\mathbf{C}}^{-1}_N (\rho)-\mathbf{\Sigma}_0^{-1}(\rho)=\widehat{\mathbf{C}}^{-1}_N (\rho)\left(\mathbf{\Sigma}_0(\rho)-\widehat{\mathbf{C}}_N (\rho)\right)\mathbf{\Sigma}_0^{-1}(\rho).
\end{equation}
along with the relation:
\begin{align*}
\mathbf{x^*}\mathbf{A}\mathbf{y}& = \textnormal{tr}\left(\mathbf{x^*}\mathbf{A}\mathbf{y}\right)\\&= 
\textnormal{vec}^*(\mathbf{x})\textnormal{vec}\left(\mathbf{Ay}\right)
\\& =\mathbf{x}^* \left(\mathbf{y^t}\otimes \mathbf{I}_N\right)\textnormal{vec}(\mathbf{A}).
\end{align*}
for ${\bf x}\in\mathbb{C}^{N\times 1},{\bf y}\in\mathbb{C}^{N\times 1}$ and ${\bf A}\in \mathbb{C}^{N\times N}$, yields
%\begin{align*}
%&\sqrt{n}\left(\mathbf{s^*_0}\widehat{\mathbf{C}}^{-1}_N (\rho)\mathbf{s_0}-\mathbf{s^*_0}\mathbf{\Sigma}_0^{-1}(\rho)\mathbf{s_0}\right) \\
%&= \sqrt{n}\: \mathbf{s^*_0}\left(\widehat{\mathbf{C}}^{-1}_N (\rho)-\mathbf{\Sigma}_0^{-1}(\rho)\right)\mathbf{s_0}.
%\end{align*}
%Using the resolvent identity:
%Also recall that 
%\begin{align*}
%\mathbf{x^*}\mathbf{A}\mathbf{y}& = \textnormal{tr}\left(\mathbf{x^*}\mathbf{A}\mathbf{y}\right)\\&= 
%\textnormal{vec}^*(\mathbf{x})\textnormal{vec}\left(\mathbf{Ay}\right)
%\\& =\mathbf{x}^* \left(\mathbf{y^t}\otimes \mathbf{I}_N\right)\textnormal{vec}(\mathbf{A}).
%\end{align*}
\begin{align*}
&\sqrt{n}\: \mathbf{s^*_0}\left(\widehat{\mathbf{C}}^{-1}_N (\rho)-\mathbf{\Sigma}_0^{-1}(\rho)\right)\mathbf{s_0}\\
& =\sqrt{n} \:\mathbf{s^*_0}\widehat{\mathbf{C}}^{-1}_N (\rho)\left(\mathbf{\Sigma}_0(\rho)-\widehat{\mathbf{C}}_N (\rho)\right)\mathbf{\Sigma}_0^{-1}(\rho)\mathbf{s_0} \\
& = \sqrt{n}\:\mathbf{s^*_0}\widehat{\mathbf{C}}^{-1}_N (\rho) \left[\left(\mathbf{s_0^t} \mathbf{\Sigma}_0^{-1}(\rho)\right) \otimes \mathbf{I}_N \right]\\
& \times \textnormal{vec}\left(\mathbf{\Sigma}_0(\rho)-\widehat{\mathbf{C}}_N (\rho)\right).
\end{align*}
Using the result of Lemma 1, we have 
\begin{align*}
\sqrt{n}\:\textnormal{vec}\left(\widehat{\mathbf{C}}_N (\rho)-\mathbf{\Sigma}_0(\rho)\right)\xrightarrow[n\to+\infty]{d} \mathbf{x} \sim \mathcal{GCN}\left(0,\mathbf{M}_1,\mathbf{M}_2\right),
\end{align*} 
where $\mathcal{GCN}\left(\mathbf{0},\mathbf{M}_1,\mathbf{M}_2\right)$ denotes the Generalized Complex Normal distribution with zero-mean, covariance matrix $\mathbf{M}_1$ and pseudo-covariance matrix $\mathbf{M}_2$.
Finally, using the following convergence relations 
\vspace{-0.15cm}
\begin{align*}
&\mathbf{s^*_0}\widehat{\mathbf{C}}^{-1}_N (\rho) \left[\left(\mathbf{s_0^t} \mathbf{\Sigma}_0^{-1}(\rho)\right) \otimes \mathbf{I}_N \right]\\
& \xrightarrow[n\to+\infty]{a.s}\mathbf{s^*_0}\mathbf{\Sigma}_0^{-1}(\rho)\left[\left(\mathbf{s_0^t} \mathbf{\Sigma}_0^{-1}(\rho)\right) \otimes \mathbf{I}_N \right]. \\
&\mathbf{s^*_0}\widehat{\mathbf{C}}^{-1}_N (\rho)\mathbf{s_0}+\mathbf{s^*_0}\mathbf{\Sigma}_0^{-1}(\rho)\mathbf{s_0} \xrightarrow[n\to+\infty]{a.s} 2 \:\mathbf{s^*_0}\mathbf{\Sigma}_0^{-1}(\rho)\mathbf{s_0}. \\
& \frac{1}{\mathbf{s^*_0}\widehat{\mathbf{C}}^{-1}_N (\rho)\mathbf{\Sigma}_N\widehat{\mathbf{C}}^{-1}_N (\rho)\mathbf{s_0}} \xrightarrow[n\to+\infty]{a.s} 
\frac{1}{\mathbf{s^*_0}\mathbf{\Sigma}_0^{-1}(\rho)\mathbf{\Sigma}_N\mathbf{\Sigma}_0^{-1}(\rho)\mathbf{s_0}}.
\end{align*}
it follows from the Slutsky's theorem \cite{stefanski} that:
\begin{equation}
\begin{split}
Q_1 \xrightarrow[n\to+\infty]{d} & \frac{-2 \:\mathbf{s^*_0}\mathbf{\Sigma}_0^{-1}(\rho)\mathbf{s_0} }{\mathbf{s^*_0}\mathbf{\Sigma}_0^{-1}(\rho)\mathbf{\Sigma}_N\mathbf{\Sigma}_0^{-1}(\rho)\mathbf{s_0}}\mathbf{s^*_0}\mathbf{\Sigma}_0^{-1}(\rho)\\ 
&\times \left[\left(\mathbf{s_0^t} \mathbf{\Sigma}_0^{-1}(\rho)\right) \otimes \mathbf{I}_N \right] \mathbf{x}
\end{split}.
\end{equation}
We now handle $Q_2$. To this end, we treat the term $\sqrt{n}\:\left(\mathbf{s^*_0}\mathbf{\Sigma}_0^{-1}(\rho)\mathbf{\Sigma}_N\mathbf{\Sigma}_0^{-1}(\rho)\mathbf{s_0}-\mathbf{s^*_0}\widehat{\mathbf{C}}^{-1}_N (\rho)\mathbf{\Sigma}_N\widehat{\mathbf{C}}^{-1}_N (\rho)\mathbf{s_0}\right)$ as follows
\vspace{-0.2cm}
\begin{align*}
&\sqrt{n}\:\left(\mathbf{s^*_0}\mathbf{\Sigma}_0^{-1}(\rho)\mathbf{\Sigma}_N\mathbf{\Sigma}_0^{-1}(\rho)\mathbf{s_0}-\mathbf{s^*_0}\widehat{\mathbf{C}}^{-1}_N (\rho)\mathbf{\Sigma}_N\widehat{\mathbf{C}}^{-1}_N (\rho)\mathbf{s_0}\right)\\
&= \sqrt{n}\:\Biggl(\mathbf{s^*_0}\mathbf{\Sigma}_0^{-1}(\rho)\mathbf{\Sigma}_N\mathbf{\Sigma}_0^{-1}(\rho)\mathbf{s_0}-\mathbf{s^*_0}\mathbf{\Sigma}_0^{-1}(\rho)\mathbf{\Sigma}_N\widehat{\mathbf{C}}^{-1}_N (\rho)\mathbf{s_0}\\ &+\mathbf{s^*_0}\mathbf{\Sigma}_0^{-1}(\rho)\mathbf{\Sigma}_N\widehat{\mathbf{C}}^{-1}_N (\rho)\mathbf{s_0}- \mathbf{s^*_0}\widehat{\mathbf{C}}^{-1}_N (\rho)\mathbf{\Sigma}_N\widehat{\mathbf{C}}^{-1}_N (\rho)\mathbf{s_0}\Biggr) \\
& = \sqrt{n}\:\mathbf{s^*_0}\mathbf{\Sigma}_0^{-1}(\rho)\mathbf{\Sigma}_N\left(\mathbf{\Sigma}_0^{-1}(\rho)-\widehat{\mathbf{C}}^{-1}_N (\rho)\right)\mathbf{s_0}  \\& 
+ \sqrt{n}\: \mathbf{s^*_0}\left(\mathbf{\Sigma}_0^{-1}(\rho)-\widehat{\mathbf{C}}^{-1}_N (\rho)\right)\mathbf{\Sigma}_N\widehat{\mathbf{C}}^{-1}_N (\rho)\mathbf{s_0}
.\end{align*}
Similarly, using the resolvent identity, we can write 
\begin{align*}
&\sqrt{n}\:\left(\mathbf{s^*_0}\mathbf{\Sigma}_0^{-1}(\rho)\mathbf{\Sigma}_N\mathbf{\Sigma}_0^{-1}(\rho)\mathbf{s_0}-\mathbf{s^*_0}\widehat{\mathbf{C}}^{-1}_N (\rho)\mathbf{\Sigma}_N\widehat{\mathbf{C}}^{-1}_N (\rho)\mathbf{s_0}\right) \\
&= \sqrt{n}\:\mathbf{s^*_0}\mathbf{\Sigma}_0^{-1}(\rho)\mathbf{\Sigma}_N \mathbf{\Sigma}_0^{-1}(\rho) \left(\widehat{\mathbf{C}}_N (\rho)-\mathbf{\Sigma}_0(\rho)\right)\widehat{\mathbf{C}}^{-1}_N (\rho)\mathbf{s_0} \\ &+ \sqrt{n}\: \mathbf{s^*_0}\mathbf{\Sigma}_0^{-1}(\rho) \left(\widehat{\mathbf{C}}_N (\rho)-\mathbf{\Sigma}_0(\rho)\right)\widehat{\mathbf{C}}^{-1}_N (\rho)\mathbf{\Sigma}_N\widehat{\mathbf{C}}^{-1}_N (\rho)\mathbf{s_0} \\
&= \sqrt{n}\: \mathbf{s^*_0}\mathbf{\Sigma}_0^{-1}(\rho)\mathbf{\Sigma}_N \mathbf{\Sigma}_0^{-1}(\rho) \left[\left(\mathbf{s_0^t}\widehat{\mathbf{C}}^{-1}_N (\rho)\right)\otimes \mathbf{I}_N\right] \\
& \times \textnormal{vec}\left(\widehat{\mathbf{C}}_N (\rho)-\mathbf{\Sigma}_0(\rho)\right) 
+ \sqrt{n}\: \mathbf{s^*_0}\mathbf{\Sigma}_0^{-1}(\rho) \\
& \times \left[\left(\widehat{\mathbf{C}}^{-1}_N (\rho)\mathbf{\Sigma}_N\widehat{\mathbf{C}}^{-1}_N (\rho)\mathbf{s_0}\right)^t \otimes \mathbf{I}_N\right]\textnormal{vec}\left(\widehat{\mathbf{C}}_N (\rho)-\mathbf{\Sigma}_0(\rho)\right)
.\end{align*}
Also note that
\begin{align*}
& \frac{\left(\mathbf{s^*_0}\mathbf{\Sigma}_0^{-1}(\rho)\mathbf{s_0}\right)^2}{\left(\mathbf{s^*_0}\widehat{\mathbf{C}}^{-1}_N (\rho)\mathbf{\Sigma}_N\widehat{\mathbf{C}}^{-1}_N (\rho)\mathbf{s_0}\right)\left(\mathbf{s^*_0}\mathbf{\Sigma}_0^{-1}(\rho)\mathbf{\Sigma}_N\mathbf{\Sigma}_0^{-1}(\rho)\mathbf{s_0}\right)}  \\
& \xrightarrow[n\to+\infty]{a.s}
\frac{\left(\mathbf{s^*_0}\mathbf{\Sigma}_0^{-1}(\rho)\mathbf{s_0}\right)^2}{\left(\mathbf{s^*_0}\mathbf{\Sigma}_0^{-1}(\rho)\mathbf{\Sigma}_N\mathbf{\Sigma}_0^{-1}(\rho)\mathbf{s_0}\right)^2}
\end{align*}
Thus, by means of Slutsky's theorem, it follows that
\begin{align*}
Q_2 \xrightarrow[n\to+\infty]{d}
& \frac{\left(\mathbf{s^*_0}\mathbf{\Sigma}_0^{-1}(\rho)\mathbf{s_0}\right)^2}{\left(\mathbf{s^*_0}\mathbf{\Sigma}_0^{-1}(\rho)\mathbf{\Sigma}_N\mathbf{\Sigma}_0^{-1}(\rho)\mathbf{s_0}\right)^2} \Biggl[\mathbf{s^*_0}\mathbf{\Sigma}_0^{-1}(\rho)\mathbf{\Sigma}_N \mathbf{\Sigma}_0^{-1}(\rho) \\
& \times \left[\left(\mathbf{s_0^t}\mathbf{\Sigma}_0^{-1}(\rho)\right)\otimes \mathbf{I}_N\right]+\mathbf{s^*_0}\mathbf{\Sigma}_0^{-1}(\rho) \\& \times \left[\left(\mathbf{\Sigma}_0^{-1}(\rho)\mathbf{\Sigma}_N\mathbf{\Sigma}_0^{-1}(\rho)\mathbf{s_0}\right)^t \otimes \mathbf{I}_N\right]\Biggr] \mathbf{x}.
\end{align*}
Gathering the convergence results of $Q_1$ and $Q_2$, we thus obtain:
\begin{align*}
\sqrt{n}\left(\widehat{\textnormal{SNR}}-\textnormal{SNR}_0\right) \xrightarrow[n\to+\infty]{d}
\mathbf{c^*}\mathbf{x}.
\end{align*}
Noticing that 
\begin{equation}
\begin{split}
\mathbf{c^*}\mathbf{x} &=  \Re(\mathbf{c})^t\Re(\mathbf{x})+\Im(\mathbf{c})^t\Im(\mathbf{c}) \\
&= \widetilde{\mathbf{c}}^t \mathbf{v},
\end{split}
\end{equation}
where $\mathbf{v}=\left[\Re(\mathbf{x})^t \Im(\mathbf{x})^t\right]^t$, it suffices thus to derive the distribution of $\mathbf{v}$. This follows from the following Lemma:
%Then, if we know the distribution of $\mathbf{v}$, we can easily derive the distribution of $\mathbf{c^*}\mathbf{x}$. In this line, the following lemma provides the distribution of the vector $\mathbf{v}$.
%\textcolor{red}{Add a link between the two results}
\begin{lemma}
Let $\mathbf{x}=\left(x_1, x_2, \cdots, x_k\right)^t$ be a zero-mean complex jointly-Gaussian random vector with covariance $\mathbf{M}_1$ and pseudo-covariance $\mathbf{M}_2$ and let $\mathbf{v}=\left[\Re(\mathbf{x})^t \Im(\mathbf{x})^t\right]^t$. Then, following the results of \cite{gallager},
\begin{equation}
\mathbf{v} \sim \mathcal{N}\left(\mathbf{0},\Xi\right)
\end{equation}
\label{convergence}
%where $\Xi=\frac{1}{2}\begin{bmatrix}
%\Re\left(\mathbf{M}_1\right)+\Re\left(\mathbf{M}_2\right) &-\Im\left(\mathbf{M}_1\right)+\Im\left(\mathbf{M}_2\right) \\ 
% \Im\left(\mathbf{M}_1\right)+\Im\left(\mathbf{M}_2\right)& \Re\left(\mathbf{M}_1\right)-\Re\left(\mathbf{M}_2\right)
%\end{bmatrix}$.
\end{lemma}
Using Lemma \ref{convergence}, we conculde that $\mathbf{c^*}\mathbf{x}$ is normally distributed with zero mean and variance $\sigma_n^2 =\widetilde{\mathbf{c}}^t \Xi \widetilde{\mathbf{c}} $. This conculdes the proof of the theorem.
\bibliographystyle{IEEEtran}
\bibliography{References}
\end{document}